\documentclass[journal,final,doublecolumn,10pt,twoside]{IEEEtranTCOM}
\IEEEoverridecommandlockouts

\normalsize

\usepackage{cite}
\usepackage{amsmath,amssymb,amsfonts,amsthm}
\usepackage{mathtools}
\usepackage{algorithmic}
\usepackage{algorithm}
\usepackage{graphicx}
\usepackage{textcomp}
\usepackage{xcolor}
\usepackage{tikz}
\usepackage{bm}
\usepackage{arydshln}
\usepackage{enumerate}  
\usepackage{bbm} 
\usepackage{mathtools}
\usepackage{physics} 
\usepackage{caption}
\usepackage{subcaption}
\usepackage[yyyymmdd]{datetime}
\usepackage{eso-pic}
\usepackage{ftnxtra}
\usepackage{fnpos}

\usetikzlibrary{matrix,positioning,chains,fit,shapes,calc,bayesnet,decorations.pathreplacing,quotes}

\usepackage[square,numbers, sort&compress]{natbib}

\def\BibTeX{{\rm B\kern-.05em{\sc i\kern-.025em b}\kern-.08em
    T\kern-.1667em\lower.7ex\hbox{E}\kern-.125emX}}
    
\theoremstyle{plain}

\theoremstyle{definition}

\newtheorem{definition}{Definition}
\newtheorem{lemma}{Lemma}

\theoremstyle{remark}

\newcommand{\etal}{\textit{et al}.}

\newcommand{\one}{\mathbbm{1}}
\newcommand{\bmx}{\bm{x}}

\DeclarePairedDelimiter\ceil{\lceil}{\rceil}

\makeatletter
\newcommand*{\rom}[1]{\expandafter\@slowromancap\romannumeral #1@}
\makeatother

\usepackage[labelformat=simple]{subcaption}

\begin{document}

\title{Using List Decoding to Improve \\the Finite-Length Performance of \\Sparse Regression Codes}

\author{Haiwen~Cao~and~Pascal O.~Vontobel,~\IEEEmembership{Fellow,~IEEE}
\thanks{Haiwen~Cao is with the Department of Information Engineering, The Chinese University of Hong Kong, Hong Kong (e-mail: ch017@ie.cuhk.edu.hk).}
\thanks{Pascal O.~Vontobel is with the Department of Information Engineering, The Chinese University of Hong Kong, Hong Kong, and also with the Institute of Theoretical Computer Science and Communications, The Chinese University of Hong Kong, Hong Kong (e-mail: pascal.vontobel@ieee.org).}
\thanks{This work was supported in part by a grant from the Research Grants Council of the Hong Kong Special Administrative Region, China, under Project CUHK 14209317.}
\thanks{This work was published in part at the IEEE Information Theory Workshop (ITW), Riva del Garda, Italy, which appears in this manuscript as reference \cite{Cao2104:Using}.}
}

\markboth{IEEE Transactions on Communications}%
{Submitted paper}

\maketitle

\begin{abstract}
We consider sparse superposition codes (SPARCs) over complex AWGN channels. Such codes can be efficiently decoded by an approximate message passing (AMP) decoder, whose performance can be predicted via so-called state evolution in the large-system limit. In this paper, we mainly focus on how to use concatenation of SPARCs and cyclic redundancy check (CRC) codes on the encoding side and use list decoding on the decoding side to improve the finite-length performance of the AMP decoder for SPARCs over complex AWGN channels. Simulation results show that such a concatenated coding scheme works much better than SPARCs with the original AMP decoder and results in a steep waterfall-like behavior in the bit-error rate performance curves. Furthermore, we apply our proposed concatenated coding scheme to spatially coupled SPARCs. Besides that, we also introduce a novel class of design matrices, i.e., matrices that describe the encoding process, based on circulant matrices derived from Frank or from Milewski sequences. This class of design matrices has comparable encoding and decoding computational complexity as well as very close performance with the commonly-used class of design matrices based on discrete Fourier transform (DFT) matrices, but gives us more degrees of freedom when designing SPARCs for various applications. 
\end{abstract}


\IEEEpeerreviewmaketitle

\section{Introduction}
\IEEEPARstart{S}{parse} superposition codes (SPARCs), also known as sparse regression codes, were first introduced by Joseph \etal~\cite{joseph2012least} for efficient communication over additive white Gaussian noise (AWGN) channels, but have also been used for lossy compression \cite{venkataramanan2014lossy, venkataramanan2014lossyperform} and multi-terminal communication \cite{venkataramanan2012sparse}. In this paper, we will only consider SPARCs for channel coding, especially for point-to-point communication. Unlike low-density parity-check (LDPC) codes with coded modulation, SPARCs directly map the message to a codeword. These codewords are formed as sparse linear combinations of the columns of a so-called design matrix. The structure of the design matrix allows one to construct low-complexity decoders with performance reasonably close to the computationally intractable maximum-likelihood decoder. Joseph \etal\ first introduced an efficient decoding algorithm called ``adaptive successive decoding'' in \cite{joseph2013fast}, and then Barron \etal~\cite{barron2012high} proposed an adaptive soft-decision successive decoder that has much better finite-length performance compared with the original successive decoder. Subsequently, a series of papers~\cite{barbier2014replica, barbier2017approximate, rush2017capacity} introduced approximate message passing (AMP) decoders, which are a class of algorithms \cite{bayati2011dynamics} approximating loopy belief propagation on dense factor graphs, for SPARCs. The adaptive soft-decision decoder in~\cite{barron2012high} and AMP decoders in~\cite{rush2017capacity} have all been proven to be asymptotically capacity-achieving when one assumes the entries of the design matrix to be i.i.d.\ samples from some zero-mean Gaussian distribution with a suitable variance. In the following, we will only consider AMP decoders for SPARCs.

The above results mainly focus on the asymptotic characterization of the error performance of SPARCs, but barely consider the finite-length performance of SPARCs. In order to improve the finite-length performance of SPARCs, Greig \etal~\cite{greig2017techniques} proposed several techniques, which include an iterative power allocation algorithm for SPARCs and concatenating LDPC codes with original SPARCs as their inner codes. These two techniques can significantly improve the finite-length performance. However, the iterative power allocation algorithm is very sensitive to code parameters as well as the channel quality, and the algorithm is controlled by some parameter related to the code rate, which can only be tuned via running extensive simulations for different values of this parameter to try to find the ``optimal'' value. This sensitivity is rather suboptimal for the realistic scenario where the channel quality is unknown or varies slowly with time. Moreover, the concatenated coding scheme only works well when the signal-to-noise ratio (SNR) is above some threshold and its performance can be even much worse than the original SPARCs when the SNR is below the threshold, although this coding scheme results in a steep waterfall-like behavior in the bit-error rate performance curves.
 
In this paper, we will mainly discuss how to tackle the above two issues, i.e., the sensitivity issue of the iterative power allocation algorithm and the degraded performance of SPARCs concatenated with LDPC codes when the SNR is below the threshold, and further improve the finite-length performance of SPARCs over complex AWGN channels. The techniques used here can be straightforwardly applied also to the real case. Besides that, we will apply this technique to spatially coupled SPARCs (SC-SPARCs), which will be discussed in Section~\ref{Extention_SC}. Moreover, we will introduce a novel class of design matrices based on circulant matrices. The contributions of this work are as follows:
\begin{enumerate}
	\item The main contribution of this paper is a concatenated coding scheme to tackle issues happening in the previous work by Greig \etal~\cite{greig2017techniques}. More specifically, we use SPARCs concatenated with CRC codes on the encoding side and propose the use of list decoding on the decoding side to further improve the finite-length performance. The improvement will be shown via simulation results. (Details about this coding scheme appear in Section~\ref{List_Decoding_Section}.)
	\item We introduce an alternative design-matrix construction, which uses a circulant matrix with a Frank sequence \cite{frank1962phase} or a Milewski sequence \cite{milewski1983periodic} as its leading row, and this alternative class of design matrices has comparable encoding and decoding computational complexity as well as very close performance with the commonly-used class of design matrices based on DFT matrices in the complex case. Moreover, it provides us with more degrees of freedom when designing SPARCs for various applications. (See Section~\ref{Alternative_Matrix_Section} for details.)
	\item As a side contribution, we propose a variant of the AMP decoder for complex AWGN channels which we have derived from a first-order approximation of some message-passing algorithm, and also provide the corresponding state evolution which is similar to the one for the real case. Due to the similarity to the real case, the iterative power allocation scheme and online estimation of parameters in the state evolution for the real case can be suitably modified to complex AWGN channels. (See Sections~\ref{Decoding} and~\ref{iterative_PA_section} for details.)
\end{enumerate}

The rest of this paper is structured as follows. We first provide some notations that will be used throughout this paper. In Section ~\ref{backgrounds}, we give some background material on SPARCs over complex AWGN channels. In Section~\ref{Alternative_Matrix_Section}, we first briefly discuss the commonly-used class of design matrices based on DFT matrices and then propose a novel class of design matrices. In Section~\ref{List_Decoding_Section}, we introduce the concatenated coding scheme for SPARCs concatenated with CRC codes (see Fig. \ref{diagram_CRC}). Moreover, simulation results are given in this section to demonstrate the significantly improved finite-length performance of the proposed coding scheme. In Section~\ref{Extention_SC}, the extension of this concatenated coding scheme to SC-SPARCs will be discussed. Finally, we conclude the paper in Section~\ref{conclusion}. Throughout this paper, the large-system limit refers to $L$, $M$, $n \to \infty$ while keeping $L\log{M} = nR$, where $L$, $M$, $n$ and $R$ will be specified in Section~\ref{Encoding}.

\subsection{Notations}\label{Notations}
We use $\log$ to denote the logarithm to the base 2, and use $\ln$ to denote the natural logarithm. We use boldface font to denote (column) vectors or matrices, plain font for scalars, and subscripts for indexing entries of a vector or a matrix. We denote the complex conjugate transpose of the matrix $\bm{A}$ and the transpose of the vector $\bm{\beta}$ by $\bm{A}^{\ast}$ and $\bm{\beta}^{\intercal}$, respectively. We denote the indicator function of a statement $\mathcal{A}$ by $\one\left(\mathcal{A}\right)$. We write $\mathcal{N}(0,\, \sigma^2)$ to denote the (real) Gaussian distribution with zero mean and variance $\sigma^{2}$. We write $\mathcal{CN}(0,\, \sigma^2)$ to denote the circularly-symmetric complex Gaussian distribution with zero mean and variance $\sigma^{2}$. For a positive integer $M$, we use $[M]$ to denote the set $\{1, \ldots, M\}$.

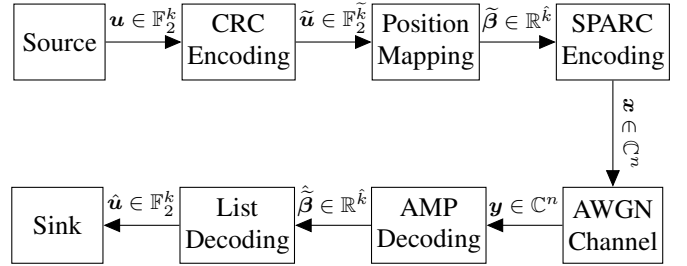
\begin{figure}
\centering 
\scalebox{1.01}{
\begin{tikzpicture}[scale = 0.82, ssnode/.style={draw, rectangle, inner sep = 2pt, minimum width = 1.2cm, minimum height = 1cm},]
\begin{scope}[start chain=going right, node distance=1cm]
\node[ssnode, on chain](v1){Source};
\node[ssnode, on chain, align=center](v2){CRC \\ Encoding};
\draw [->] (v1) -- (v2) node[midway, above, sloped]{\footnotesize $\bm{u} \in \mathbb{F}_2^{k}$};
\node[ssnode, on chain, align=center](v3){Position \\ Mapping};
\draw [->] (v2) -- (v3) node[midway, above, sloped]{\footnotesize $\bm{\widetilde{u}} \in \mathbb{F}_2^{\widetilde{k}}$};
\node[ssnode, on chain,  align=center](v4){SPARC \\ Encoding};
\draw [->] (v3) -- (v4) node[midway, above, sloped]{\footnotesize $\widetilde{\bm{\beta}} \in \mathbb{R}^{\hat{k}}$};
\end{scope}
\begin{scope}[start chain = going below, node distance=1.4cm]
\chainin (v4);
\node[ssnode, on chain,  align=center](v5){AWGN \\ Channel};
\draw[->] (v4) -- (v5) node[midway, above, sloped]{\footnotesize $\bm{x} \in \mathbb{C}^{n}$};
\end{scope}

\begin{scope}[start chain = going left, node distance=1cm]
\chainin (v5);
\node[ssnode, on chain,  align=center](v6){AMP \\ Decoding};
\draw[->] (v5) -- (v6) node[midway, above, sloped]{\footnotesize $\bm{y} \in \mathbb{C}^{n}$};
\node[ssnode, on chain,  align=center](v7){List \\ Decoding};
\draw[->] (v6) -- (v7) node[midway, above, sloped]{\footnotesize $\hat{\widetilde{\bm{\beta}}} \in \mathbb{R}^{\hat{k}}$};
\node[ssnode, on chain,  align=center](v8){Sink};
\draw [->] (v7) -- (v8) node[midway, above, sloped]{\footnotesize $\hat{\bm{u}} \in \mathbb{F}_2^{k}$};
\end{scope}
\end{tikzpicture}
}
\caption{Block diagram of a communication system employing SPARCs combined with CRC codes as outer-error detection codes. The parameters in the block diagram satisfy the following relationships: $k=L\cdot \log{M}, \widetilde{k} = \widetilde{L} \cdot \log{M}, \hat{k} = \widetilde{L} \cdot M$.} \label{diagram_CRC}
\end{figure}

\section{SPARCs and AMP Decoders} \label{backgrounds}
\subsection{Construction and Encoding} \label{Encoding}
Encoding of a SPARC is defined in terms of a design matrix $\bm{A}$ of size $n \times ML$, where $n$ is the block length and where $M$ and $L$ are positive integers that are specified below in terms of $n$ and the rate $R$. In the original construction of SPARCs,  the entries of $\bm{A}$ are assumed to be i.i.d.\ Gaussian $\sim \mathcal{N}(0,\,1/n)\,$. For this paper, the entries of $\bm{A}$ have been extended to be i.i.d.\ Gaussian $\sim \mathcal{CN}(0,\,1/n)\,$ for the complex AWGN channels. 

Codewords are constructed as sparse linear combinations of the columns of $\bm{A}$. Specifically, a codeword $\bm{x} = \left(x_1, \ldots, x_n\right)$ is of the form $\bm{A} \bm{\beta}$, where $\bm{\beta} = \left(\bm{\beta}_{1}^{\intercal}, \ldots, \bm{\beta}_{L}^{\intercal} \right)^{\intercal}$ is an $ML \times 1$ vector and where each section $\bm{\beta}_{\ell} \triangleq \left(\beta_{(\ell-1)\cdot M+1}, \ldots, \beta_{\ell \cdot M}\right)^{\intercal}, \ell=1, \ldots, L$, has the property that only one of its components is non-zero. The non-zero value\footnote{In this paper, we mainly focus on techniques for improving the performance of SPARCs, so we only encode information by the location of non-zero entries in $\bm{\beta}$ with values fixed a priori for simplicity. However, the values (specifically the phases) of the non-zero entries can also be different in the case of complex channels, and the details w.r.t.\ these so-called modulated SPARCs can be found in \cite{hsieh2020modulated}.} of  section $\bm{\beta}_{\ell}$ is set to $\sqrt{nP_{\ell}}$, where $P_{1}, \ldots, P_{L}$ are pre-specified positive constants that satisfy $\sum\limits_{\ell=1}^{L} P_{\ell} = P$. This guarantees that $\frac{1}{n} \sum\limits_{i=1}^{n} \abs{x_i}^2 \leq P$ with high probability. (See Appendix B in the arXiv version \cite{barron2010least} of \cite{joseph2012least} for details.)

Both the matrix $\bm{A}$ and the power allocation $\{ P_{1}, \ldots, P_{L}\}$ are known to the sender and the receiver before transmission. The choice of power allocation plays a crucial role in determining the finite-length performance of SPARCs, which will be illustrated in the following. Without loss of generality, we can choose the power allocation to be non-increasing since messages are independent across sections and the AWGN channels we consider here are memoryless. 

Because the design matrix $\bm{A}$ can be regarded as $L$ blocks with $M$ columns each, and because we pick one column from each block of the design matrix to comprise one codeword, the total number of codewords for a SPARC will be $M^{L}$. With the block length being $n$, the rate $R$ will be $\frac{\log(M^{L})}{n}$, i.e., $\frac{L\cdot \log M}{n}$. In actual implementations, the input bitstream is split into chunks of $\log M$ bits, and each chunk of $\log M$ bits is first bijectively mapped to the position of the non-zero element for the corresponding section. Based on $L$ consecutive chunks, the sender constructs the message vector $\bm{\beta}$ by assigning the pre-specified positive constants to the corresponding positions for each section (see ``position mapping'' in Fig. \ref{diagram_CRC}). Eventually, the sender transmits the codeword $\bm{x} \triangleq \bm{A}\bm{\beta}$ (see ``SPARC encoding'' in Fig. \ref{diagram_CRC}).

\subsection{AMP Decoders for SPARCs over Complex AWGN Channels} \label{Decoding}
For decoding SPARCs over complex AWGN channels, we derive the following AMP decoder, which is similar to the one used for the real case in \cite{rush2017capacity}.

Namely, given the channel output $\bm{y} \triangleq \bm{A}\bm{\beta} + \bm{w}$, where $\bm{w} = \left(w_i\right)_{i \in [n]}$ and $w_i$ are i.i.d.\ $\mathcal{CN}(0,\, \sigma^2)$ for all $i \in [n]$, the AMP decoder will generate successive estimates of the message vector, denoted by $\{ \bm{\beta}^{t}\}, \bm{\beta}^{t} \in \mathbb{R}^{LM}$, for $t =0, 1, 2, \ldots$.

Initialize $\bm{\beta}^{0} \coloneqq  \bm{0}$, then compute
\begin{align}
\bm{z}^{t} &\coloneqq  \bm{y} - \bm{A}\bm{\beta}^{t} + \frac{\bm{z}^{t-1}}{\tau_{t-1}^{2}} \left(P - \frac{\norm{\bm{\beta}^{t}}^2}{n}\right), \\
\beta_{i}^{t+1} &\coloneqq  \eta_{i}^{t}\left(\bm{\beta}^{t} + \bm{A}^{\ast}\bm{z}^{t}\right), \text{ } i = 1, \ldots, ML, \label{denoiser_part}
\end{align}
where the constants $\{\tau_t\}$ and the denoiser functions $\eta_{i}^{t}\left(\cdot\right)$ are defined as follows. First, define 

\begin{align}
\tau_{0}^{2} \coloneqq  \sigma^2 + P, \quad  \tau_{t+1}^2 \coloneqq  \sigma^2 + P\cdot\left(1 -  x\left(\tau_{t}\right)\right), \quad t \geq 0, \label{SE_1}
\end{align}
where
\begin{align}
&x\left(\tau\right) \coloneqq  \nonumber \\ \!&\sum_{\ell=1}^{L}\!\frac{P_{\ell}}{P}\! \cdot \!\mathbb{E}\!\left[\frac{e^{2\cdot\frac{\sqrt{nP_{\ell}}}{\tau}\cdot\left(\Re{U_{1}^{\ell}}+\frac{\sqrt{nP_{\ell}}}{\tau}\right)}}{e^{2\cdot\frac{\sqrt{nP_{\ell}}}{\tau}\!\cdot\!\left(\!\Re{U_{1}^{\ell}}+\frac{\sqrt{nP_{\ell}}}{\tau}\!\right)}\!+\!\sum_{j=2}^{M}e^{2\cdot\frac{\sqrt{nP_{\ell}}}{\tau}\cdot \Re{U_{j}^{\ell}}}}\right] \!\label{SE_2}, 
\end{align}
where $\{U_{j}^{\ell}\}$ are i.i.d.\ $\mathcal{CN}(0,\, 1)$, for all $j \in [M]$ and all $\ell \in [L]$. Note that Equations \eqref{SE_1} and \eqref{SE_2} together are called the state evolution. Finally, for all $\ell \in [L]$ and all $i \in \mathrm{sec}_{\ell} \triangleq \{(\ell-1)\cdot M+1, \ldots, \ell\cdot M\}$, define
\begin{align}
\eta_{i}^{t}\left(\bm{s}\right) \triangleq  \sqrt{nP_\ell} \cdot \frac{e^{2\cdot \Re{s_i}\cdot\frac{\sqrt{nP_\ell}}{\tau_{t}^{2}}}}{\sum_{j \in \mathrm{sec}_{\ell}}e^{2\cdot \Re{s_j}\cdot\frac{\sqrt{nP_\ell}}{\tau_{t}^{2}}}}. \label{denoisers}
\end{align}

The above AMP decoder and its state evolution for SPARCs over complex AWGN channels is the same as those in the real case \cite{rush2017capacity} up to a factor 2 appearing in \eqref{denoisers} and \eqref{SE_2} and a few minor modifications, which is because we can regard transmission of SPARCs with rate $R$ over a complex AWGN channel with a given SNR as two independent (orthogonal) transmissions of SPARCs with rate $\frac{R}{2}$ over real AWGN channels with the same SNR. 

The denoiser functions \eqref{denoisers} for this AMP decoder can be derived from the minimum mean squared error (MMSE) estimator of $\beta_{i}$, for $i \in [ML]$, which is Bayes-optimal under the distributional assumption on $\bm{s} = \bm{\beta} + \tau\bm{u}$ with $\bm{u}$ having i.i.d.\ $\mathcal{CN}(0,\, 1)$ entries and independent with $\bm{\beta}$. I.e., $\eta_{i}^{t}\left(\bm{s}\right)  \triangleq \mathbb{E} \left[\beta_{i} | \bm{s} = \bm{\beta} + \tau\bm{u} \right]$, where the expectation is over $\bm{\beta}$ and $\bm{u}$ given $\bm{s}$.

\subsection{Iterative Power Allocation and Online Estimation of $\tau_t$} \label{iterative_PA_section}
From \eqref{SE_1}, we see that the effective noise variance $\tau_{t}^2$ is the sum of two terms, where the first term is the channel noise variance $\sigma^2$, and the other term $P\cdot\left(1 -  x\left(\tau_{t-1}\right)\right)$ can be regarded as the interference from the undecoded sections in $\bm{\beta}^{t}$. In other words, $ x\left(\tau_{t-1}\right)$ is the expected fraction of sections that are correctly decoded at the end of iteration $t$. In the following, we adapt Lemma 1 in \cite{greig2017techniques}, which gives upper and lower bounds on $x\left(\tau\right)$ for the real AWGN case, to our complex AWGN case. This adaptation requires the redefinition of $\nu_{\ell}$ from \cite[Lemma 1]{greig2017techniques}.

\begin{lemma} \label{bounds_x}
Let $\nu_{\ell} \triangleq \frac{2LP_\ell}{R\tau^2 \ln{2}}$. For sufficiently large $M$, and for any $\delta \in (0, \frac{1}{2})$, 
\begin{align}
x(\tau) &\leq \sum_{\ell=1}^{L}\frac{P_{\ell}}{P}\left[\one\{\nu_{\ell} > 2-\delta\} + M^{-\kappa_{1}\delta^{2}} \one\{\nu_{\ell} \leq 2-\delta\} \right],\\
x(\tau) &\geq \left(1-\frac{M^{-\kappa_{2}\delta^{2}}}{\delta\sqrt{\ln M}}\right)\sum_{\ell=1}^{L}\frac{P_{\ell}}{P}\one\{\nu_{\ell} > 2+\delta\},
\end{align}
where $\kappa_{1}, \kappa_{2}$ are universal positive constants.
\end{lemma}
\begin{proof}
The proof is similar to that of \cite[Lemma 1]{greig2017techniques} and is therefore omitted.
\end{proof}

In the limit $M \to \infty$, we can use the following approximation for $x(\tau)$: 
\begin{align}
x(\tau) \approx \sum_{\ell=1}^{L}\frac{P_{\ell}}{P}\one\bigl\{LP_{\ell} > R\tau^2\ln{2}\bigr\}. \label{approx_power}
\end{align}	
This approximation for $x(\tau)$ gives us a quick way to check whether a given power allocation scheme leads to reliable decoding in the large-system limit. Based on this approximation, we can design an iterative power allocation algorithm as follows. Namely, the $L$ sections of the SPARC are divided into $B$ blocks of $L/B$ sections each, and the same power is allocated to every section within a block. We sequentially allocate the power to all blocks in the following way: for each section within the first block, we allocate the minimum power needed so that all sections within this block can be decoded in the first iteration when $\tau_0^{2} \coloneqq \sigma^2 + P$. Using the approximation \eqref{approx_power}, we assign the power 
\begin{align} 
P_\ell \coloneqq \frac{R\tau_0^2\ln{2}}{L}, \quad 1 \leq \ell \leq \frac{L}{B}, \label{assign_pa}
\end{align}
to each section within the first block, and then we get $\tau_1^{2} \coloneqq \sigma^2 + (P- \frac{L}{B} \cdot P_1)$ due to \eqref{SE_1}. We sequentially allocate the power derived in the same way as \eqref{assign_pa} to the remaining blocks. For $R \leq C$,\footnote{Here, $C \triangleq \log\left(1+\frac{P}{\sigma^2}\right)$ is the channel capacity.} we can readily derive the result that $\sum\limits_{\ell=1}^{L} P_{\ell}$ will be less than the average power $P$. In order to fully allocate the average power $P$, we can slightly modify the algorithm in the following way: for every block, we first compare the average of the remaining available power with the minimum required power computed similarly to \eqref{assign_pa}. If the former one is larger than the latter one, we allocate the remaining power to the remaining sections evenly and terminate the algorithm; otherwise, we assign the minimum required power to each section within the current block and the algorithm continues. The above iterative power allocation scheme for the complex case is straightforwardly extended from the one for the real case proposed in~\cite{greig2017techniques}.

When we consider the finite-length performance of SPARCs, it has been observed that the above iterative power allocation algorithm works better than the exponentially-decaying power allocation which was proved in \cite{rush2017capacity} to be asymptotically capacity-achieving in the large-system limit. Since the above iterative scheme is derived based on the asymptotic approximation for $x(\tau)$ (see Eq.~\eqref{approx_power}), the power allocated by using the above iterative scheme may be slightly different from the non-asymptotic bounds on $x(\tau)$ in Lemma~\ref{bounds_x} at finite lengths. In order to compensate their difference, we can introduce a parameter $R_{\mathrm{PA}}$ that serves as the code rate $R$ in \eqref{assign_pa}; by carefully tuning this parameter to be slightly different from the code rate, we can further improve the finite-length performance. The details of this iterative power allocation scheme for the real case can be found in \cite{greig2017techniques} and all the reasoning can be straightforwardly modified to the complex AWGN channel scenario considered here.

The effective noise variance $\tau_{t}^{2}$ can be estimated via \eqref{SE_1}, \eqref{SE_2} in advance for iteration $t$ up to the maximum iteration $T$,\footnote{The number of iterations can also be determined in advance by running the state evolution until it converges to some fixed point (within a specified tolerance).} but this procedure is extremely time-consuming since we need to estimate $L$ expectations over $M$ random variables via Monte-Carlo simulation for each iteration $t$. Rush \etal\ in \cite{rush2018error} have already shown that $\tau_{t}^{2}$ is concentrated around $\frac{\norm{\bm{z}^{t}}^2}{n}$, i.e., we can estimate 
\[
\hat{\tau}_{t}^{2} = \frac{\norm{\bm{z}^{t}}^2}{n}.
\]
We call this an online estimate of $\tau_{t}^{2}$, and Greig \etal~\cite{greig2017techniques} have already empirically shown that this online estimate provides a good estimate of the noise variance in each iteration as it accurately reflects how the decoding is progressing in that simulation run.

\section{An alternative class of design matrices for SPARCs over Complex AWGN Channels} \label{Alternative_Matrix_Section}
In theoretical analyses of SPARCs, the entries of the design matrix $A$ are i.i.d.\ samples from a normal distribution with zero mean. With such a matrix, the computational complexity of matrix-vector multiplication\footnote{The matrix-vector multiplications we consider in this paper include $\bm{A} \bm{\beta}$ and $\bm{A}^{\ast}\bm{z}$, where the former multiplication is used for both encoding and decoding, whereas the latter multiplication is only used for decoding.} is $O(LMn)$ and the space complexity is prohibitive for reasonable code lengths. 

In order to reduce these complexities, most previous papers used a class of design matrices based on DFT matrices (subsequently denoted by $A_\mathrm{D}$) to replace the original class of design matrices.
(Note that the commonly-used design matrix is based on a Hadamard matrix in the real case.) In this paper, we will introduce an alternative design-matrix construction based on circulant matrices (subsequently denoted by $A_\mathrm{C}$) whose leading rows are derived from Frank or from Milewski sequences.
(Note that for the real case, the commonly-used design matrix is based on a Hadamard matrix and we can also introduce a circulant matrix with a leading row based on the well-known m-sequences.)
In the following, we will first illustrate the DFT-based design-matrix construction and then show how to construct the novel class of design matrices based on circulant matrices in detail.

Before introducing these classes of design matrices, we need to discuss the basic criteria for picking such design matrices. Note that properties of the original class of design matrices discussed in Section~\ref{Encoding} include near orthogonality between pairs of columns, near orthogonality between pairs of rows, row sums close to zero, and column sums close to zero. Mimicking these properties of the original design matrices, we would like to construct matrices such that pairs of columns are as orthogonal as possible, and also such that the column sums and the row sums are close to zero. 

A first DFT-matrix based approach to construct design matrices is based on taking a random subset of rows of a single large DFT matrix.\footnote{Note that the first all-one row and the first all-one column need to be removed from the DFT matrix because they do not satisfy the requirements of column sums and row sums being close to zero.} One readily sees that DFT matrices satisfy the above requirements. Let us discuss the advantages of such a class of design matrices based on DFT matrices over the original class of design matrices regarding complexities. Since the class of design matrices based on DFT matrices is constructed via picking a row-permuted DFT matrix with suitable size, the only information we need to store is this row permutation, which significantly reduces the memory required. Besides that, we can perform the matrix-vector multiplications via a fast Fourier transform (FFT) to greatly reduce the computational complexity. However, the matrix-vector multiplication directly based on such a large DFT matrix is very inefficient because $ML$ is usually far greater than $n$. (Recall that the design matrix has size $n\times ML$.) 

\begin{figure*}[h]
\normalsize
\setcounter{equation}{9}
\begin{align}
A_{\mathrm{C}}  \triangleq 
\left[
\renewcommand{\arraystretch}{2}
\begin{array}{c:c:c:c:c}
\pi_{1,1}(\phi_{1}\phi'_{1}C) & \pi_{1,2}(\phi_{1}\phi'_{2}C) & \quad \cdots \quad  & \pi_{1,L-1}(\phi_{1}\phi'_{L-1}C) & \pi_{1,L}(\phi_{1}\phi'_{L}C)
\\ \hdashline[2pt/2pt]
\pi_{2,1}(\phi_{2}\phi'_{1}C) & \pi_{2,2}(\phi_{2}\phi'_{2}C) & \cdots & \pi_{2,L-1}(\phi_{2}\phi'_{L-1}C) & \pi_{2,L}(\phi_{2}\phi'_{L}C)
\\ \hdashline[2pt/2pt]
\vdots & \vdots & \ddots & \vdots & \vdots
\\ \hdashline[2pt/2pt]
\pi_{L_{\mathrm{BR}},1}(\phi_{L_{\mathrm{BR}}}\phi'_{1}C) & \pi_{L_{\mathrm{BR}},2}(\phi_{L_{\mathrm{BR}}}\phi'_{2}C) & \cdots & \pi_{L_{\mathrm{BR}},L-1}(\phi_{L_{\mathrm{BR}}}\phi'_{L-1}C) & \pi_{L_{\mathrm{BR}},L}(\phi_{L_{\mathrm{BR}}}\phi'_{L}C)
\end{array}
\right],  \label{circulant_matrix}
\end{align}

\hrulefill
\vspace*{4pt}
\end{figure*}

A second DFT-matrix based approach to construct design matrices is based on defining 
\[
A_{\mathrm{D}} \triangleq
\left[
\begin{array}{c;{2pt/2pt}c;{2pt/2pt}c}
A_{\mathrm{D}_{1}} & \cdots & A_{\mathrm{D}_{L}}
\end{array}
\right],
\]
where for $i \in [L]$, $A_{\mathrm{D}_{i}}$ is a (truncated) row-permuted version of the DFT matrix with suitable size. Importantly, the row permutations are taken independently. Note that the size of $A_{\mathrm{D}_{i}}$ is $n \times M$ for all $i \in [L]$. However, a drawback of this DFT-matrix based approach is that, in order to have radix-2 fast Fourier transforms available, we need to start with a DFT matrix of size $2^k \times 2^k$ instead of size $\max(n+1, M+1) \times \max(n+1, M+1)$, where $k = \ceil{\log(\max(n+1, M+1))}$,\footnote{The $n+1$ and $M+1$ instead of $n$ and $M$ in the definition of $k$ come from removing the first all-one row and removing first all-one column in the DFT matrix, respectively.} and this results in the growth of the encoding and the decoding computational complexity.

Although the class of design matrices based on DFT matrices works well in terms of efficiency and performance, it lacks flexibility since DFT matrices are fixed. It is better to have more degrees of freedom for the design matrix when designing SPARCs for various applications, as these degrees of freedom allow one to satisfy other desirable constraints. Next, we will introduce an alternative design-matrix construction based on circulant matrices to reach such flexibility. Namely, we choose the design matrix $A_\mathrm{C}$ to be  an $L_{\mathrm{BR}} \times L$ array (where $L_{\mathrm{BR}} \triangleq \ceil{\frac{n}{M}}$ and where $L$ is an even number) of row-permuted circulant matrices of size $M \times M$ of the form shown in \eqref{circulant_matrix} at the top of the page, and all the parameters are introduced as follows: $\pi_{j,i}$ is a row permutation for $ j \in [L_{\mathrm{BR}}], i \in [L]$, $C$ is a circulant matrix with a carefully-chosen leading row; $\{\phi_{j} \in \mathbb{C}\ | \ j \in [L_{BR}]\}$ should be chosen to satisfy $\sum\limits_{j=1}^{L_{BR}} \phi_{j} = 0$ and also $\abs{\phi_j} =1$ for $j \in [L_{\mathrm{BR}}]$. For example, $\{\phi_{j} \in \mathbb{C}\ | \ j \in [L_{BR}]\}$ can be chosen as the $L_{BR}$-roots of unity. In the same way, $\{\phi'_{i} \in \mathbb{C}\ | \ i \in [L]\}$ should be chosen to satisfy $\sum\limits_{i=1}^{L} \phi'_{i} = 0$ and also $\abs{\phi'_i} =1$ for $i \in [L]$. For example, $\{\phi'_{i} \in \mathbb{C}\ | \ i \in [L]\}$ is chosen as $\{\phi'_{i} = (-1)^{i-1} \ | \ i \in [L]\}$ in our simulations (see Fig.~\ref{Differ_Matrices}). The main goal of the above construction is to get a design matrix with zero row sums, zero column sums and near orthogonality between pairs of columns. The multiplication of such a circulant-based design matrix by a vector can be done efficiently with the suitable use of FFTs, inverse FFTs (IFFTs), and permutations. Note that the complexities with respect to storage and computation are comparable with the class of design matrices based on DFT matrices, but this novel class of design matrices based on circulant matrices can provide more degrees of freedoms.

Let us discuss the leading row of the circulant matrix $C$ in more detail. The requirement for the leading row is to make the correlation between it and its cyclically shifted versions as low as possible. For sequences over complex numbers, there are so-called ``perfect sequences'' that have all of their nontrivial periodic autocorrelations equal to zero. Particular examples include Frank sequences~\cite{frank1962phase} and Milewski sequences~\cite{milewski1983periodic}. Let us give the definitions of these two different sequences in the following.

\begin{definition}[Frank sequences]
Let $d$ be a positive integer. A Frank sequence $\bm{\theta}$ of length $d^2$ is defined by 
\[
\theta_{j+kd+1} \triangleq \exp\left(\frac{2\pi i jk}{d}\right),
\]
where $j$ and $k$ are integers satisfying $0 \leq j,\! k < d$, and $i$ is the imaginary unit.
\end{definition}

\begin{definition}[Milewski sequences]
Let $d$ be a positive integer and let $h$ be a non-negative integer. A Milewski sequence $\bm{\theta}$ of length $d^{2h+1}$ is defined by 
\[
\theta_{j+kd^h+1} \triangleq 
\begin{cases}
	\exp\left(\frac{\pi ik \left(2j+kd^h\right)}{d^{h+1}}\right) & \text{for even } d \\
	\exp\left(\frac{\pi ik \left(2j+\left(k+1\right)d^h\right)}{d^{h+1}}\right) & \text{for odd } d
\end{cases},
\]
where $j$ and $k$ are integers satisfying $0 \leq j < d^h$ and $0 \leq k < d^{h+1}$, and $i$ is the imaginary unit.
\end{definition}

The above two families of sequences have already been proven to be perfect in \cite{frank1962phase, milewski1983periodic}. Taking one perfect sequence with length $M$ as the leading row, we can construct a circulant matrix with uncorrelated rows as well as columns. The resulting circulant matrix is the circulant matrix $C$ used to construct the right-hand side of \eqref{circulant_matrix}.\footnote{This construction is similar to Gallager's construction of regular LDPC codes in \cite{gallager1962low}.} 

\begin{figure}[t!]
\centering
\begin{subfigure}{\linewidth}
\centering
\includegraphics[scale=0.18]{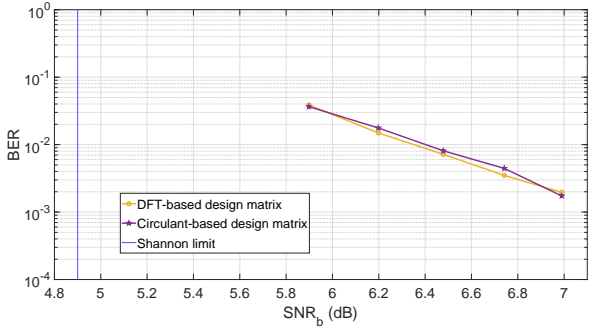}
\caption{BER performance comparison of SPARCs with $L=1024, M=512, R=1.8$ bits/(channel use)/dimension over a complex AWGN channel.} 
\end{subfigure}
\hfill
\begin{subfigure}{\linewidth}
\centering
\includegraphics[scale=0.18]{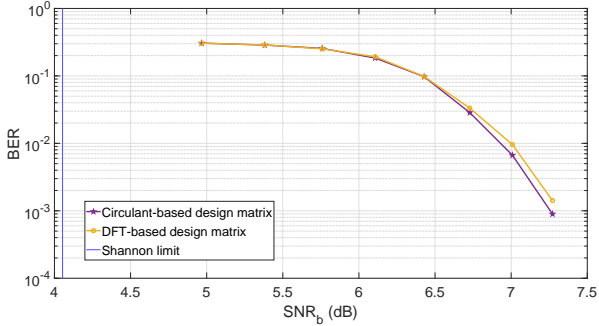}
\caption{BER performance comparison of $\left(w=6, L_{C}=32\right)$ SC-SPARCs with $L=960, M=128$ over a complex AWGN channel.} 
\end{subfigure}
\caption{BER performance comparison of (SC-)SPARCs with different design matrices.}  \label{Differ_Matrices}
\end{figure}

Simulation results (see Fig. \ref{Differ_Matrices}) for the two discussed classes of design matrices show that they lead to comparable performances, with the class of design matrices based on circulant matrices being slightly better for the SC-SPARCs. (SC-SPARCs will be introduced in Section~\ref{Extention_SC}.) Besides these comparable performances, the class of design matrices based on circulant matrices~\eqref{circulant_matrix} has more degrees of freedom in the sense that we can choose any $\{\phi_{j}\in \mathbb{C} \ | \  j \in [L_{BR}]\}$ satisfying $\sum\limits_{j=1}^{L_{BR}} \phi_{j} = 0$ and $\abs{\phi_j} =1$ for $j \in [L_{BR}]$ and also choose any $\{\phi'_{i}\in \mathbb{C} \ | \  i \in [L]\}$ satisfying $\sum\limits_{i=1}^{L} \phi'_{i} = 0$ and $\abs{\phi'_i} =1$ for $i \in [L]$. This allows one to satisfy other desirable constraints when designing SPARCs for various applications, e.g., the magnetic recording scenario.

For example, in magnetic recording systems, a characteristic of digital recording channels is that they suppress the low frequency components of the recorded data, thus codes are required to have large rejection of low frequency components (see \cite[Chapter 19]{vasic2004coding} for details). Such codes are called ``spectrum shaping codes''. Based on the structure of our circulant-based design matrix, we can easily verify that $\sum\limits_{i=1}^{n} x_i = 0$ where $\bmx = A_\mathrm{C}\cdot\bm{\beta}$. Such a code is called DC-free and belongs to the class of spectrum shaping codes. 

\section{Using List Decoding to improve the Performance of SPARCs} \label{List_Decoding_Section}
While running simulations for evaluating the performance of the AMP decoder for original SPARCs, we noticed that a significant amount of wrongly decoded sections can be decoded successfully if we choose the second or third most likely location based on the output of the AMP decoder. This observation inspired us to consider list decoding, i.e., instead of outputting one location for each section, the decoding algorithm should output $S$ candidates for each section based on the estimated vector $\bm{\beta}^{(T)}$ from the AMP decoder, where $S$ is the size of the output list and can be pre-specified. In order to improve the performance of this list decoder, we use an outer code that mainly serves as an error detection code. Details are explained in the following subsection.

\subsection{Encoding and Decoding SPARCs with CRC codes} \label{CRC_encoding_decoding}
The key part of our proposed concatenated coding scheme is using CRC codes as outer codes. There are two ways to employ CRC codes: (a) an ``inter-section-based approach" that encodes the original message to generate extra check sections bit by bit when we focus on the bit error rate (BER) performance, (b) an ``intra-section-based approach" that encodes them section by section when we focus on the section error rate (SecER) performance.  We will only discuss the former way since the latter way is essentially the same.\footnote{The only difference between these two SPARCs with CRC codes is the CRC encoding part, and there is no need to discuss the section-by-section encoding separately.} The block diagram in Fig.~\ref{diagram_CRC} gives a high-level view of the proposed concatenated coding scheme with SPARCs as inner codes and CRC codes as outer codes, where CRC codes with suitable parameters were taken from \cite{koopman2004cyclic}. In the following, we discuss the main blocks of Fig.~\ref{diagram_CRC}.

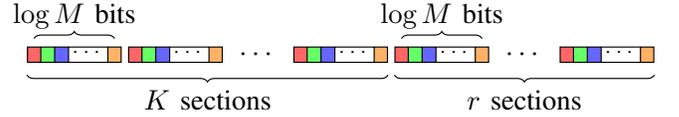
\begin{figure}
\centering
\scalebox{1.07}
{
\begin{tikzpicture}[start chain=going right,node distance=.5ex]
  \tikzset{
    simple node/.style={
      draw,
      text height=0.8ex,text depth=0.4ex,
      inner sep=0pt,text width=1ex,align=center
    },
    split node/.style={
      simple node,
      rectangle split,rectangle split horizontal,rectangle split parts=5,
	rectangle split part fill={red!60, green!60, blue!60, white!60, orange!60},
      draw,inner sep=0ex,rectangle split part align=base,
    },
  }
  \node[on chain,split node] (v1){\nodepart{two}\nodepart{three}\nodepart[text width=3ex]{four}{\scriptsize $\cdots$} \nodepart{five}};
  \node[on chain,split node] {\nodepart{two}\nodepart{three}\nodepart[text width=3ex]{four} {\scriptsize $\cdots$} \nodepart{five}};
  \node[on chain] {$\cdots$} ;
  \node[on chain,split node] (v2){\nodepart{two}\nodepart{three}\nodepart[text width=3ex]{four} {\scriptsize $\cdots$} \nodepart{five}};
   \draw[decorate,decoration={brace}] ($(v1)+(-0.5,0.2)$) --
        node[above]{$\log{M}$ bits} ++(1,0) ; 
   \path (v1.south west) 
		edge[decorate,decoration={brace,mirror,raise=.15cm},"$K$ sections" below=6pt](v1.south west -| v2.south east);
  \node[on chain,split node] (v3){\nodepart{two}\nodepart{three}\nodepart[text width=3ex]{four}{\scriptsize $\cdots$} \nodepart{five}};
  \node[on chain] {$\cdots$} ;
  \node[on chain,split node] (v4){\nodepart{two}\nodepart{three}\nodepart[text width=3ex]{four} {\scriptsize $\cdots$} \nodepart{five}};
   \draw[decorate,decoration={brace}] ($(v3)+(-0.5,0.2)$) --
        node[above]{$\log{M}$ bits} ++(1,0) ; 
   \path (v3.south west) 
		edge[decorate,decoration={brace,mirror,raise=.15cm},"$r$ sections" below=6pt](v3.south west -| v4.south east);
\end{tikzpicture}
}
\caption{This diagram illustrates how to encode $K$ sections of information bits to generate $r$ sections of check bits. More specifically, we encode $K$ information bits with the same color to generate $r$ check bits with the same color via CRC encoding.} \label{CRC_bits}
\end{figure}

(\textbf{CRC Encoding}) We first encode the original message $\bm{u}$ with size $L\cdot  \log{M}$ bit by bit in the following way: 
\begin{enumerate}
	\item Partition the original message $\bm{u}$ into $L$ sections, denoted by $\{\bm{u}_1, \ldots, \bm{u}_L\}$. These $L$ sections are organized into $\frac{L}{K} \left(\triangleq N_{\mathrm{g}}\right)$ groups of $K$ sections each.\footnote{For simplicity, we assume that $L$ is divisible by $K$. For the case that $L$ is not divisible by $K$, we can take the last mod$\left(L, K\right)$ sections as one group, and encode this group individually.} More precisely, these $N_{\mathrm{g}}$ groups will be $\{\bm{u}_1, \bm{u}_{N_{\mathrm{g}}+1},\ldots, \bm{u}_{L-N_{\mathrm{g}}+1}\}$, $\ldots$, $\{\bm{u}_j, \bm{u}_{N_{\mathrm{g}}+j}, \ldots, \bm{u}_{L-N_{\mathrm{g}}+j}\}$, $\ldots$, $\{\bm{u}_{N_{\mathrm{g}}}, \bm{u}_{2N_{\mathrm{g}}}, \ldots, \bm{u}_{L}\}$.
	\item Use the CRC code to systematically encode each group of $K$ sections bit by bit as shown in Fig.~\ref{CRC_bits} to generate $r$ extra sections that are appended at the end of the original message; for instance, the group $\{\bm{u}_j, \bm{u}_{N_{\mathrm{g}}+j}, \ldots, \bm{u}_{L-N_{\mathrm{g}}+j}\}$ will generate extra sections $\{  \widetilde{\bm{u}}_{L+j}, \widetilde{\bm{u}}_{L+N_{\mathrm{g}}+j}, \ldots, \widetilde{\bm{u}}_{L+(r-1) \cdot N_{\mathrm{g}}+j} \}$.
\end{enumerate}

This CRC encoding procedure yields $L + N_{\mathrm{g}} \cdot r\ \bigl(\triangleq \widetilde{L}\bigr)$ sections. One can readily see that $K$ is the number of information bits and $r$ is the number of check bits for each codeword of CRC encoding, and $N_{\mathrm{g}} \cdot \log{M} \bigl(\triangleq N_{\mathrm{C}}\bigr)$ is the total number of codewords; the resulting codewords are denoted by $\{\mathrm{C}_i | i \in [N_{\mathrm{C}}]\}$. The number of additional parity bits $N_{\mathrm{C}}\cdot r$ are the cost for error detection, which leads to a trade-off between error detection capability and rate loss.

(\textbf{Position Mapping} and \textbf{SPARC Encoding}) We first map the encoded message to the corresponding $\widetilde{\bm{\beta}}$ and then perform SPARC encoding by multiplying the design matrix with $\widetilde{\bm{\beta}}$ to get the corresponding codeword $\bm{x}$, which is transmitted over the complex AWGN channel. 

(\textbf{AMP Decoding} and \textbf{List Decoding}) At the receiver side, we can decode the received message in the following way:
\begin{enumerate}
	\item Perform $T$ iterations of AMP decoding;
	the resulting estimate of $\widetilde{\bm{\beta}}$ is called $\widetilde{\bm{\beta}}^{(T)}$.
	\item For each section $\ell \in \bigl[ \widetilde{L}\bigr]$, normalizing $\widetilde{\bm{\beta}}^{(T)}_{\ell}$ gives the a posterior distribution estimate of the location of the non-zero entry of $\widetilde{\bm{\beta}}_{\ell}$, denoted by $\hat{\tilde{\boldsymbol{\beta}}}^{(T)}_{\ell}$.
	\item  For each section $\ell \in \bigl[ \widetilde{L}\bigr]$, convert the posterior distribution estimate $\hat{\tilde{\boldsymbol{\beta}}}^{(T)}_{\ell}$ into $\log_2{M}$ bit-wise posterior distribution estimates according to Algorithm~\ref{convertSec2Bits}, which is essentially the same as ``Algorithm 2" in \cite{greig2017techniques}.
	\item For each codeword $\mathrm{C}_i$, we establish a binary tree of depth $K+r$, where, starting at the root, at each layer, we keep  at most $S$ branches, which are the most likely ones.
	\item For each codeword $\mathrm{C}_i$, once we have established such a binary tree, list decoding will give us $S$ ordered candidates corresponding to the remaining $S$ paths from the root to the leaves. For each path (i.e., a candidate with $K+r$ bits) from the most likely one to the least likely one, we detect whether this candidate is valid or not via CRC detection and take the candidate as $\hat{\mathrm{C}}_i$ once the CRC condition is satisfied.\footnote{Although we have already found the most likely valid codeword, it might be an ``undetected error" in the sense that the codeword is different from what we transmitted. This is the worst case since we cannot even realize that we made an error.} If no candidate satisfies the CRC condition after checking all candidates, it means that there is a ``detected error" and we take the first candidate as $\hat{\mathrm{C}}_i$. Although at least one bit is decoded incorrectly, this first candidate can be potentially better than other candidates.
\end{enumerate}

\begin{algorithm}[t]
\caption{Conversion of the section-wise posterior distribution estimate into the bit-wise posterior distribution estimate\\  (using Matlab notation) }
\label{convertSec2Bits}
\begin{algorithmic}
\REQUIRE section-wise posterior distribution estimate $\hat{\bm{\beta}}_{\mathrm{sec}} = \bigl(\hat{\beta}_1, \ldots, \hat{\beta}_M\bigr)$.
\ENSURE the probability of each bit being $0$ $\{\mathrm{b}_i, \forall i \in \left[\log{M}\right]\}$.
\STATE  $M_{\text{bits}} \leftarrow \log{M}$
\STATE  $\text{pos} \leftarrow 1 : M$
\FOR{\text{bits $\leftarrow 1 : M_{\text{bits}}$}}
	\STATE $n_{c} \leftarrow 2^{\text{bits}-1}$
	\STATE $\text{pos}_{r} \leftarrow$ reshape$\bigl(\text{pos}, n_{c}, \frac{M}{n_{c}}\bigr)$
	\STATE $\text{pos}_{0} \leftarrow$ reshape$\bigl(\text{pos}_{r}\left(\,: \, , 1:2:\text{end}\right), 1, \frac{M}{2}\bigr)$ 
	\STATE $b_{\text{bits}} \leftarrow$ sum$\bigl(\hat{\bm{\beta}}_{\mathrm{sec}}\bigl(\text{pos}_{0}\bigr)\bigr)$
\ENDFOR

\RETURN $\{b_{i}, \forall i \in \left[M_{\text{bits}}\right]\}$

\end{algorithmic}
\end{algorithm}

Besides the above regular list decoding assisted with CRC codes, we can try to further improve the performance by running AMP and list decoding again (denoted by ``\textbf{AMP again}") for parts of the message, especially when most of the errors stemming from list decoding are detected errors. More specifically, we can apply the following procedure:
\begin{enumerate}
	\item Run AMP decoding as before, except that at each iteration, fix the ``correctly decoded''\footnote{It may contain wrongly decoded sections which were regarded as correct ones; we call these errors ``undetected errors''. (See also Footnote 10.)} parts of the message and only estimate the other sections. When the maximum number of iteration $T$ is reached or some halting condition is satisfied, the algorithm outputs $\widetilde{\bm{\beta}}^{\ast}$.
	\item Take only the wrongly decoded sections of $\widetilde{\bm{\beta}}^{\ast}$, denoted by $\widetilde{\bm{\beta}}^{\ast}_{\mathrm{WD}}$, and apply the above list decoding procedure to $\widetilde{\bm{\beta}}^{\ast}_{\mathrm{WD}}$, which gives the decoded message.
\end{enumerate} 

\subsection{Numerical Simulations} \label{simulation_SPARCs_PA}
In this subsection, we evaluate the BER performance of SPARCs concatenated with CRC codes and SPARCs without CRC codes over the complex AWGN channel for different overall rates. For SPARCs concatenated with CRC codes, we consider the finite-length performance of list decoding with different list sizes as well. For complex AWGN channels, the SNR per information bit, i.e., $\mathrm{SNR}_{\mathrm{b}}$, is defined as $\frac{P}{\sigma^2}\cdot \frac{1}{R}$.

The setup for the simulation results in Fig. \ref{low_rate} is as follows. We consider SPARCs with overall rate $R=0.8$ bits/(channel use)/dimension, in which case $R_{\mathrm{PA}}=0$ and the iterative power allocation scheme gives a flat allocation (see Section~\ref{iterative_PA_section} for details). Moreover, 
\begin{itemize}
\item we choose the number of information sections $L$ to be 1000,
\item we choose the size of each section $M$ to be 512,
\item we choose the number of information bits $K$ to be 100,
\item we use the 8-bit CRC code whose generator polynomial is 0x97$=x^8+x^5+x^3+x^2+x+1$ (see, e.g., \cite{koopman2004cyclic}).
\end{itemize}

\begin{figure}[t!]
\centering
\begin{subfigure}{\linewidth}
\centering
\includegraphics[scale=0.25]{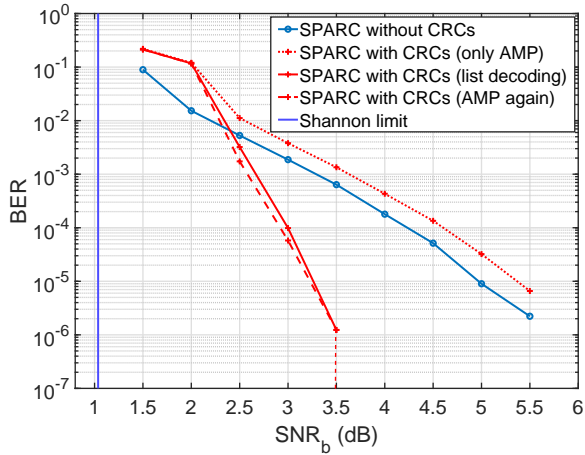}
\caption{BER performance comparison of SPARCs with CRC codes using list decoding and original SPARCs without CRC codes using only AMP. Besides that, ``AMP again'' is also included.}  \label{SPARCs_ListD_low_rate}
\end{subfigure}
\hfill
\begin{subfigure}{\linewidth}
\centering
\includegraphics[scale=0.19]{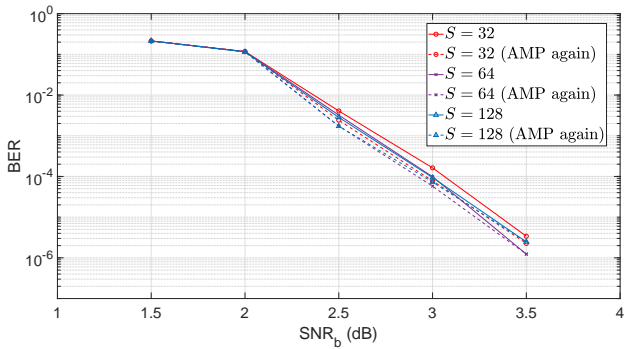}
\caption{BER performance comparison of SPARCs with CRC codes for different list sizes. The solid lines show performance results for one round of list decoding, whilst the dashed lines show performance results of ``AMP again".  } \label{differListSize_low_rate}
\end{subfigure}
\caption{BER performance comparison of SPARCs with overall rate $R=0.8$ bits/(channel use)/dimension.}  \label{low_rate}
\end{figure}

The difference between Fig. \ref{SPARCs_ListD_low_rate} and Fig. \ref{differListSize_low_rate} is as follows. Namely, Fig. \ref{SPARCs_ListD_low_rate} considers the list decoder with the list size $S=64$, whereas Fig. \ref{differListSize_low_rate} considers the performance of SPARCs concatenated with CRC codes under list decoding with different list sizes. The plot in Fig. \ref{SPARCs_ListD_low_rate} shows that SPARCs concatenated with CRC codes can provide a steep waterfall-like behavior\footnote{Note that the dashed vertical line in Fig. \ref{SPARCs_ListD_low_rate} indicates that no error was observed in $10^3$ simulation runs.} above a threshold of $\mathrm{SNR}_{\mathrm{b}} = 3.5$ dB. Besides this, the proposed concatenated coding scheme can also provide a boost compared with the original SPARCs below this threshold. This is different from SPARCs with LDPC codes as outer codes that were considered in \cite{greig2017techniques}, since the performance of that concatenated coding scheme is dramatically worse than the original SPARCs below the threshold. Note that in the case of concatenation with LDPC codes \cite{greig2017techniques}, the ``AMP again" procedure improves the performance significantly. However, in our case, running  ``AMP decoder again" usually yields rather limited improvements. We can analyze this as follows.
\begin{itemize}
\item When $\mathrm{SNR}_{\mathrm{b}} \lesssim 2\text{dB}$, the AMP decoder can barely decode a few sections and the number of errors are beyond the CRC code's error detection capability; therefore, the output of list decoding may contain several ``undetected errors", which misleads further decoding in the ``AMP again" part.
\item When  $2\text{dB} \lesssim \mathrm{SNR}_{\mathrm{b}} \lesssim 3.5\text{dB}$, the outer error-detection code improves the performance and results in very few ``undetected errors" since parts of the received message with few errors are within the CRC code's error detection capability. Because there are very few ``undetected errors", we can further improve the performance by runing the ``AMP again" procedure. 
\item When $\mathrm{SNR}_{\mathrm{b}} \gtrsim 3.5$dB, almost all the errors are corrected when performing list decoding and there are very few errors remaining; therefore, usually there is no noticeable improvement obtained by running ``AMP again".
\end{itemize}

Fig. \ref{differListSize_low_rate} illustrates how the list size affects the performance of SPARCs with CRC codes. From this plot we deduce that $S = 64$ is the best choice for the considered setup. This can be explained as follows. The list size cannot be too small, otherwise we cannot find a candidate satisfying the CRC condition, and the resulting estimate will be the same as the estimate deduced from the original AMP decoder. However, the list size also cannot be too large since it may result in a few ``undetected errors", which will mislead the ``AMP again" part.

\begin{figure}[t!]
\centering
\includegraphics[scale=0.25]{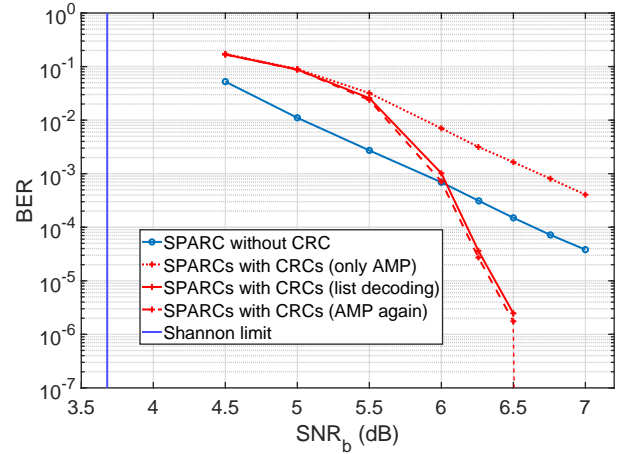}
\caption{BER performance comparison of SPARCs with CRC codes using list decoding and original SPARCs without CRC codes using only AMP. Besides that, ``AMP again'' is also included.}  \label{high_rate}
\end{figure}

The setup of the simulation results in Fig.~\ref{high_rate} is the same as for Fig.~\ref{low_rate}, except that we consider SPARCs with overall rate $R=1.5$ bits/(channel use)/dimension and $R_{\mathrm{PA}} = 3 \times 1.1 = 3.3$ $(\approx 3)$ in this case. Unsurprisingly, the plot in Fig.~\ref{high_rate} exhibits the same waterfall-like behavior\footnote{Note that the dashed vertical line in Fig.~\ref{high_rate} indicates that no error was observed in $10^3$ simulation runs.} as in Fig.~\ref{SPARCs_ListD_low_rate}. Based on the non-increasing power allocation assumption, we know that all errors appear in the last few sections and may even be consecutive. Such consecutive errors can typically be dealt with successfully thanks to the interleaving feature of the grouping scheme proposed in Section \ref{CRC_encoding_decoding}. Whereas the original SPARC is very sensitive with respect to the choice of $R_{\mathrm{PA}}$, our coding scheme turns out to be relatively robust to the choice of $R_{\mathrm{PA}}$ due to our choice of outer error-detection code and decoding scheme. Therefore, we tackled the sensitivity issue of the iterative power allocation scheme to a great extent thanks to our concatenated coding scheme.

\begin{figure}[t!]
\centering
\begin{subfigure}{\linewidth}
\centering
\includegraphics[scale=0.19]{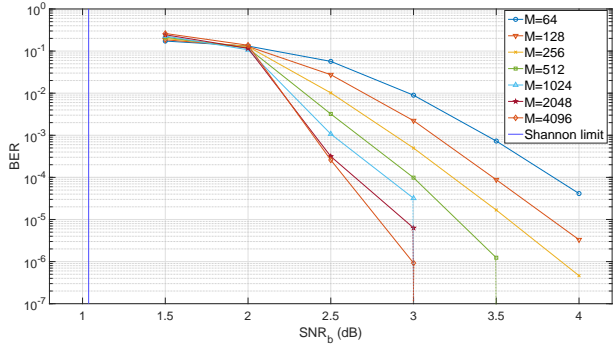}
\caption{BER performance comparison of SPARCs concatenated with CRC codes using list decoding when the overall rate is $R=0.8$ bits/(channel use)/dimension.}  \label{SPARCs_differM_low_rate}
\end{subfigure}
\hfill
\begin{subfigure}{\linewidth}
\centering
\includegraphics[scale=0.19]{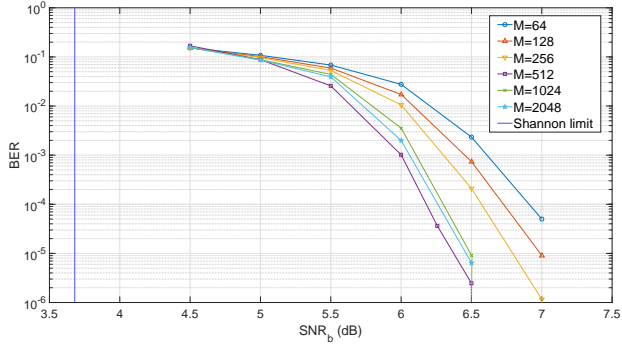}
\caption{BER performance comparison of SPARCs concatenated with CRC codes using list decoding when the overall rate is $R=1.5$ bits/(channel use)/dimension.} \label{SPARCs_differM_high_rate}
\end{subfigure}
\caption{BER performance comparison of SPARCs concatenated with CRC codes using list decoding for different section sizes $M$.}  \label{SPARCs_differM}
\end{figure}

The simulation results regarding the prediction for the {SecER}\footnote{The BER performance should be roughly half of the SecER performance since the position of the non-zero value is chosen uniformly for each section.} via state evolution in~\cite{greig2017techniques} shows that the performance gets better when the section size $M$ gets larger. In the following, we investigate how the section size $M$ affect the BER performance of our concatenated coding scheme. The setup of the simulation results in Fig.~\ref{SPARCs_differM} are the same as for the previous ones except that we consider how the section size $M$ will affect the BER performance of our concatenated coding scheme. Specifically, Fig.~\ref{SPARCs_differM_low_rate} and Fig.~\ref{SPARCs_differM_high_rate} illustrate the BER performance comparison of SPARCs concatenated with CRC codes using list decoding when the overall rate is $R=0.8$ bits/(channel use)/dimension and when the overall rate is $R=1.5$ bits/(channel use)/dimension, respectively. When the overall rate is $R=0.8$ bits/(channel use)/dimension, Fig.~\ref{SPARCs_differM_low_rate} shows that the BER performance becomes better when the section size $M$ gets larger; especially, when the section size $M$ is larger than or equal to $1024$, these SPARCs concatenated with CRC codes have a better threshold. This result coincides with the prediction. However, when the overall rate is $R=1.5$ bits/(channel use)/dimension, Fig.~\ref{SPARCs_differM_high_rate} shows that there is a ``best'' section size, i.e., $M^{*}=512$; in other words, the BER performance of our concatenated coding scheme with the section size $M^{*}=512$ is the best for section sizes from $M=2^6$ to $M=2^{11}$. This result can be explained as follows; in the context of high-rate SPARCs (say, the overall rate $R > 1$), as $M$ goes beyond the ``best'' $M^{*}$, the further benefit from the extra increase in $M$ is weaker than the loss of concentration on the prediction via state evolution. The effect of the section size $M$ on concentration is still theoretically unclear and the prediction of the ``best'' section size $M^{*}$ even for the original SPARC is still an open problem.

After discussing the simulation results, let us discuss the encoding and decoding complexities regarding our concatenated coding scheme in the following. The encoding complexity is mainly dominated by the inner encoding, i.e., one vector-matrix multiplication, which can be done efficiently via one FFT. The complexity of one FFT is $O(LM\log(LM))$, thus the encoding complexity is $O(LM\log(LM))$. The decoding complexity is mainly dominated by AMP decoding and list decoding. Firstly, the complexity of AMP decoding is dominated by two FFTs, so its complexity is $O(LM\log(LM))$. Secondly, the complexity of list decoding is $O(\frac{\widetilde{L}}{K} \log(M) \left(2S\cdot (K+r) + S\right))$, i.e., $O(S L  \log(M))$. Therefore, the decoding complexity is $O(LM\log(LM)+SL\log(M))$. Once the CRC code for our concatenated coding scheme is chosen, the number of redundant bits $r$ will be constant and it will stay constant as $L$ grows large. The number of information bits $K$ can be varied within some range (see, e.g., \cite{koopman2004cyclic}). The choice of $K$ will not only affect the rate loss of the overall code, but also affect the error detection capability of the CRC code. Thus, the choice of $K$ leads to a trade-off between error detection capability and rate loss. Furthermore, the choice of CRC codes, i.e., the parameters $r$ and $K$, will not affect the decoding complexity since the decoding complexity is $O(LM\log(LM)+SL\log(M))$.

\section{Extensions to Spatially Coupled SPARCs} \label{Extention_SC}
\subsection{Related Work}
Besides AMP decoders for the original SPARCs having been proven to be asymptotically capacity achieving with a suitably chosen power allocation, ``spatial coupling" is the other technique with which the corresponding AMP decoder has been proved to be asymptotically capacity achieving as well. The ``spatial coupling" technique was first introduced in the context of LDPC codes. In particular, \cite{kudekar2011threshold,kudekar2013spatially} showed that the ``coupling" of several copies of individual codes increases the algorithmic threshold (belief propagation (BP) threshold) of this new ensemble to the information-theoretic threshold (maximum a posteriori (MAP) threshold) of the underlying ensemble. This phenomenon is called ``threshold saturation'' and it has also been observed in the compressed sensing setting~\cite{kudekar2010effect, donoho2013information}. In the context of SPARCs, spatial coupling was first introduced by Barbier~\etal in~\cite{barbier2015approximate} and then Barbier~\etal~\cite{barbier2017approximate} used the replica analysis to show that SC-SPARCs with AMP decoders are capacity achieving over AWGN channels. The fully rigorous and complete proof was first given by Rush~\etal~\cite{rush2020capacity}. SC-SPARCs can be suitably modified to the complex AWGN channel, and Hsieh~\etal~\cite{hsieh2020modulated} proposed modulated (spatially coupled) SPARCs in which information is encoded by both the location and the value of the non-zero entries in $\bm{\beta}$ in the case of complex AWGN channels. Since we mainly focus on techniques for improving the finite-length performance, SC-SPARCs for complex AWGN channels here will only encode information via the localation of the non-zero entries in $\bm{\beta}$. In other words, we will consider unmodulated SC-SPARCs over complex AWGN channels as in \cite{hsieh2020modulated}. The details regarding AMP decoders for SC-SPARCs over complex AWGN channels can be found in \cite{hsieh2020modulated} and will be omitted here.

\subsection{List Decoding for Spatially Coupled SPARCs}
The previous results mainly focus on the asymptotic characterization of the error performance of SC-SPARCs. There are very few papers (e.g.,~\cite{liang2017clipping}) discussing how to further improve the finite-length performance of SC-SPARCs. In this subsection, we will show that list decoding can improve the finite-length performance of SC-SPARCs for the low-rate region\footnote{When the overall rate of our concatenated coding scheme is less than 1 bit/(channel use)/dimension, we say it is in the low-rate region; otherwise, we say it is in the high-rate region.} via simulation results. A detailed introduction of SC-SPARCs can be found in~\cite{hsieh2018spatially}.

Before we get into the list decoding part, we need first introduce some extra parameters used in the design matrix of SC-SPARCs via the following definition (see \cite{hsieh2018spatially} for details).

\begin{definition}[A $\left( w, \Lambda \right)$ base matrix for SC-SPARCs]
The $\left( w, \Lambda \right)$ base matrix $W$ is described by two parameters: coupling width $w \geq 1$ and coupling length $\Lambda \geq 1$. The matrix has $L_R = \Lambda + w -1$ rows and $L_C = \Lambda$ columns with exactly $w$ non-zero entries each column. Specifically, for an average power constraint $P$, the $(r, c)$th entry of the base matrix is given by 
\[
W_{r,c} = 
	\begin{cases}
		P \cdot \frac{\Lambda+w -1}{w} \!\! & \text{if}\ c \leq r \leq c+w -1, \\
		0 \!\! & \text{otherwise.}
	\end{cases}
\! r \in \left[L_R\right], c \in \left[L_C\right]. 
\]
\end{definition}

For example, for $w = 3$, the above $\left(w, \Lambda \right)$ base matrix $W$ will be in the shape 
\[
W = 
\begin{bmatrix}
    W_{1,1} & 0 & 0 & \cdots  & 0 \\
    W_{2,1} & W_{2,2} & 0 & \cdots  &  0\\
    W_{3,1} & W_{3,2} & W_{3,3} & \cdots  & 0 \\
    0 &  W_{4,2}& W_{4,3} & \ddots & 0 \\
    0& 0  & W_{5,3} & \ddots  & W_{L_{R}-2,L_C} \\
     \vdots & \vdots  & \vdots & \ddots  & W_{L_{R}-1,L_C}\\
     0&0  &0  & \cdots  & W_{L_R,L_C}
\end{bmatrix}.
\]

Replacing each non-zero entry with an $M_R \times M_C$  matrix with entries i.i.d. sampled from $\mathcal{CN}(0,\,W_{r,c}/L)$ and each zero entry with an $M_R \times M_C$ all-zero matrix, we can finally get the design matrix $A$ with the same structure as the above base matrix. Note that $n = M_R  L_R$ and $ML=M_C L_C$.

The procedure for encoding and decoding SC-SPARCs with CRC codes is very similar to the one discussed in Section~\ref{CRC_encoding_decoding}, so it is omitted here. We evaluate the BER performance and the SecER performance of SC-SPARCs concatenated with CRC codes and without CRC codes over the complex AWGN channels for different overall rates.

In order to make the performance of SC-SPARCs and the performance of the original SPARCs discussed previously comparable, we will consider the settings as close as possible to the ones in Section~\ref{simulation_SPARCs_PA}.\footnote{Due to the structure of design matrices of SC-SPARCs, the overall rate of the concatenated code cannot be chosen arbirarily.}

The setup for the simulation results in Fig.~\ref{SC_low_rate} is as follows. We consider SC-SPARCs with overall rate $R = 0.8$ bits/(channel use)/dimension. Moreover,
\begin{itemize}
\item we choose the number of information sections $L$ to be 1000;
\item we choose the size of each section $M$ to be 512;
\item we choose the number of information bits $K$ to be 100;
\item we use the 8-bit CRC code whose generator polynomial is 0x97$=x^8+x^5+x^3+x^2+x+1$ (see, e.g., \cite{koopman2004cyclic});
\item for the base matrix, we choose $\Lambda = 40$, $w = 6$; therefore, $L_C = \Lambda = 40$, $L_R = \Lambda + w -1 =45$;
\item we choose the list size for list decoding to be $S=64$;
\end{itemize}

\begin{figure}[t!]
\centering
\begin{subfigure}{\linewidth}
\centering
\includegraphics[scale=0.19]{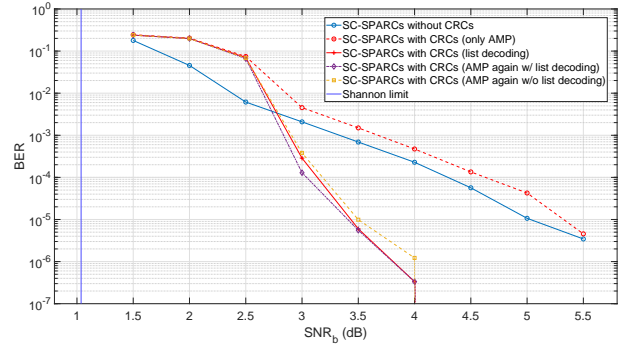}
\caption{BER performance comparison of SC-SPARCs with CRC codes using list decoding and original SC-SPARCs without CRC codes using only AMP. Besides that, ``AMP again'' with and without using list decoding are also included.}  \label{SC_BER_low_rate}
\end{subfigure}
\hfill
\begin{subfigure}{\linewidth}
\centering
\includegraphics[scale=0.19]{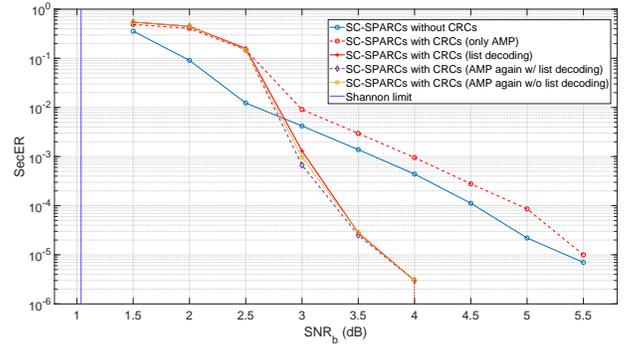}
\caption{SecER performance comparison of SC-SPARCs with CRC codes using list decoding and original SC-SPARCs without CRC codes using only AMP. Besides that, ``AMP again'' with and without using list decoding are also included.} \label{SC_SecER_low_rate}
\end{subfigure}
\caption{BER and SecER performances comparison of SC-SPARCs with overall rate $R=0.8$ bits/(channel use)/dimension.}  \label{SC_low_rate}
\end{figure}

The plot in Fig. \ref{SC_BER_low_rate} shows that SC-SPARCs concatenated with CRC codes can also provide a steep waterfall-like behavior above a threshold of $\mathrm{SNR}_{\mathrm{b}} = 4$ dB, which is larger than the result in the original SPARCs case. Furthermore, we observe that the BER performance of SC-SPARCs when ``AMP again" is employed without further list decoding is worse than the BER performance when ``AMP again" is not employed. This unusual result can be explained as follows, based on the SecER performance of SC-SPARCs.
\begin{itemize}
    \item From Fig. \ref{SC_SecER_low_rate}, we see that the SecER performance of SC-SPARC when ``AMP again" is employed without further list decoding is better than the SecER performance when ``AMP again" is not employed. This explains that the AMP decoder for the remaining part can further decode a few sections of original messages.
    \item However, in terms of BER performance, the amount of extra bits successfully decoded by the AMP decoder for the remaining part might not be larger than the amount of corrected bits within the wrongly decoded sections during the list decoding stage. This is because list decoding separately decodes the original messages based on the way we partition the original messages during the CRC encoding procedure and this decoding results in some mostly corrected sections which have corrected most of bits contained in each section's information. (E.g., when $M=512$, each section contains 9 bits of information.) Because the AMP decoder decodes the original messages in the section-wise level, the AMP decoder cannot guarantee that it decodes more bits successfully although it successfully decodes some sections within these remaining sections. 
\end{itemize}

\begin{figure}[t!]
\centering
\begin{subfigure}{\linewidth}
\centering
\includegraphics[scale=0.19]{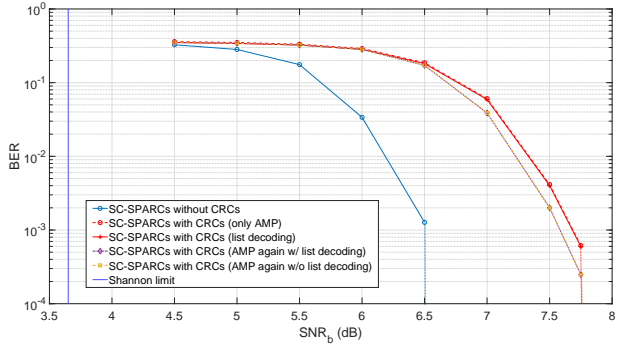}
\caption{BER performance comparison of SC-SPARCs with CRC codes using list decoding and original SC-SPARCs without CRC codes using only AMP. Besides that, ``AMP again'' with and without using list decoding are also included.}  \label{SC_BER_high_rate}
\end{subfigure}
\hfill
\begin{subfigure}{\linewidth}
\centering
\includegraphics[scale=0.19]{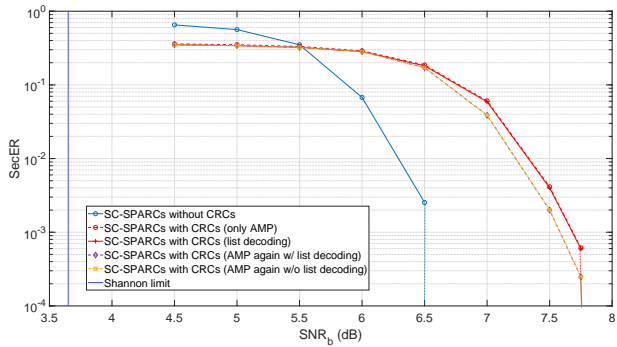}
\caption{SecER performance comparison of SC-SPARCs with CRC codes using list decoding and original SC-SPARCs without CRC codes using only AMP. Besides that, ``AMP again'' with and without using list decoding are also included.} \label{SC_SecER_high_rate}
\end{subfigure}
\caption{BER and SecER performances comparison of SC-SPARCs with overall rate $R=1.493$ bits/(channel use)/dimension.}  \label{SC_high_rate}
\end{figure}

The setup of the simulation results in Fig. \ref{SC_high_rate} is the same as for Fig. \ref{SC_low_rate} except that we consider SC-SPARCs with overall rate $R = 1.493$ bits/(channel use)/dimension. Surprisingly, the performance of this concatenated coding scheme is much worse than the performance of SC-SPARCs without CRC codes. Therefore, we conclude that under the current setup, list decoding for SC-SPARCs with CRC codes in the high-rate region might not significantly improve the finite-length performance. It is an interesting direction for future work to take the spatial coupling structure into account when we encode the original messages with CRC codes to further improve the finite-length performance of SC-SPARCs concatenated with CRC codes especially for the high-rate region. Another observation from Fig.~\ref{SC_SecER_high_rate} is that the SecER performance of SC-SPARCs with CRC codes is better than the SecER performance of SC-SPARCs without CRC codes when $\mathrm{SNR}_{\mathrm{b}}$ is below some value (in this case, $\mathrm{SNR}_{\mathrm{b}} < 5.5$ dB), and this is because our implementation scheme puts all the CRC redundant sections in the middle of the encoded messages. Due to the wave-like decoding behavior of SC-SPARCs, those CRC redundant sections are more likely to be wrongly decoded and relatively more information sections will be decoded successfully. This finally results in a relatively lower SecER since SecER is the number of section errors within information sections divided by the number of information sections.

\section{Conclusion} \label{conclusion}
In this paper, we considered SPARCs over complex AWGN channels. We first proposed the AMP decoder for complex AWGN channels, and then introduced a novel design-matrix construction based on circulant matrices. Finally, we proposed a concatenated coding scheme that uses SPARCs concatenated with CRC codes on the encoding side and uses list decoding on the decoding side. The finite-length performance of this scheme is significantly improved compared with the original SPARCs as well as the coding scheme in \cite{greig2017techniques}. Moreover, our concatenated coding scheme has the additional benefit of exhibiting insensitivity of parameters used in the iterative power allocation scheme, especially for high-rate SPARCs. Besides that, we further applied this concatenated coding scheme to SC-SPARCs. However, currently, it does not seem to be a good choice based on the same setup as the original SPARCs case.

Some interesting directions for future work regarding list decoding and (spatially coupled) SPARCs are as follows.
\begin{itemize}
    \item The concatenated coding scheme proposed in this paper can be naturally extended to modulated SPARCs proposed by Hsieh~\etal in \cite{hsieh2020modulated}. 
    \item It is natural to ask how the list size affects the performance, so one of the interesting directions will be to conduct an (information) theoretical analysis of our list decoding scheme here. (E.g., \cite{cocskun2020bounds} provides some information-theoretic quantities associated with the list size required for successive-cancellation-based list decoding of polar codes.) 
    \item Recently, several papers applied SPARCs and AMP decoders in unsourced random access scenario, e.g., \cite{fengler2019sparcs, amalladinne2020approximate,amalladinne2020unsourced}. The structure of SC-SPARCs should be a good fit for this scenario especially if we combine it with list decoding suitably (e.g., \cite{liang2020compressed}). 
    \item The performance of our concatenated coding scheme for SC-SPARCs was worse than the performance in original SPARCs case with power allocation scheme under the current setup. We expect that it can be largely improved if we combine the power allocation and spatial coupling in a suitable way.
\end{itemize}

\ifCLASSOPTIONcaptionsoff
  \newpage
\fi

\bibliographystyle{IEEEtranTCOM}
\bibliography{IEEEabrv,mybib}

\end{document}